\documentclass[a4paper,draft]{amsart}
\usepackage[a4paper,margin=25mm]{geometry}
\usepackage{amsmath}
\usepackage{amssymb,tikz,url}

\parskip\smallskipamount

\let\set\mathbb
\def\<#1>{\langle#1\rangle}

\def\ord{\operatorname{ord}}

\def\Wr{\operatorname{Wr}}
\def\vec#1{\boldsymbol{#1}}

\usepackage{mathtools}
\mathtoolsset{showonlyrefs, showmanualtags}

\newtheorem{thm}{Theorem}

\newtheorem{lem}[thm]{Lemma}
\newtheorem{defi}[thm]{Definition}

\newtheorem{ex}[thm]{Example}

\usepackage[textwidth=2cm]{todonotes}

\begin{document}

\author[Manuel Kauers]{Manuel Kauers}
\address{Manuel Kauers, Institute for Algebra, J. Kepler University Linz, Austria}
\email{manuel.kauers@jku.at}

\author[Christoph Koutschan]{Christoph Koutschan}
\address{Christoph Koutschan, RICAM, Austrian Academy of Sciences, Linz, Austria}
\email{christoph.koutschan@ricam.oeaw.ac.at}

\author[Thibaut Verron]{Thibaut Verron}
\address{Thibaut Verron, Institute for Algebra, J. Kepler University Linz, Austria}
\email{thibaut.verron@jku.at}

\thanks{M.\ Kauers was supported by the Austrian FWF grants 10.55776/PAT8258123, 10.55776/PAT9952223,
  and 10.55776/I6130.
  C.\ Koutschan was supported by the Austrian FWF grant 10.55776/I6130.
T.\ Verron was supported by the Austrian FWF grant 10.55776/P34872.}  

\title{A Shape Lemma for Ideals of Differential Operators}

\begin{abstract}
  We propose a version of the classical shape lemma for zero-dimensional
  ideals of a commutative multivariate polynomial ring to
  the noncommutative setting of zero-dimensional ideals in an algebra of
  differential operators.
\end{abstract}

\maketitle

\section{Introduction}

In the classical theory of Gr\"obner bases for commutative polynomial
rings~\cite{buchberger65,cox92,becker93,cox05,buchberger10},
the shape lemma makes a statement about the form of the Gr\"obner basis with
respect to a lexicographic term order of an ideal of dimension zero.
It was proposed by Gianni and Mora~\cite{gianni89}, and it is almost obvious.

Consider an ideal $I\subseteq K[x,y]$ in a commutative polynomial ring over
a perfect field~$K$. The ideal has dimension zero if and only if the corresponding
algebraic set
\[
V(I)=\{\,(\xi,\eta)\in \bar K^2\mid\forall\ p\in I: p(\xi,\eta)=0\,\}
\]
is finite. Here, $\bar K$ denotes the algebraic closure of~$K$.

The finitely many points in $V(I)$ have only finitely many distinct $x$-coordinates,
and if $p$ is a generator of the elimination ideal $I\cap K[x]$, then the roots of
$p$ are precisely these $x$-coordinates.
The shape lemma says that usually there is another polynomial $q\in K[x]$ with
$\deg(q)<\deg(p)$ such that $I$ is generated by $\{y-q,p\}$.
This $q$ is then the interpolating polynomial of the points in~$V(I)$.

There may be no ideal basis of the required form if $V(I)$ contains two distinct
points with the same $x$-coordinate. The ideal is said to be \emph{in normal position}
(w.r.t.~$x$) if this is not the case, i.e., if any two distinct elements of $V(I)$
have distinct $x$-coordinates. If $K$ is sufficiently large, then every ideal $I$
of dimension zero can be brought into normal position by applying a change of variables.

\begin{thm}\label{thm:1} (cf. Prop. 3.7.22 in~\cite{kreuzer00}). Let $P=K[x,y_1,\dots,y_n]$,
  let $I\subseteq P$ be an ideal of dimension zero,
  let $t=\dim_K P/I$, and suppose that $|K|>\binom t2$. Then
  there are constants $c_1,\dots,c_n\in K$ such that mapping $x$ to
  $x+c_1y_1+c_2y_2+\cdots+c_ny_n$ (and each $y_i$ to itself)
  transforms $I$ into an ideal in normal position w.r.t.~$x$.
\end{thm}

A basis of the required form may also fail to exist if $I$ is not a radical
ideal. Recall that for a radical ideal~$I$, we have $\dim_K K[x,y]/I=|V(I)|$. Also
recall that if $p$ is a generator of $I\cap K[x]$, then the equivalence classes
$[1],[x],\dots,[x^{\deg p-1}]$ are linearly independent over $K$
and $[1],[x],\dots,[x^{\deg p-1}],[x^{\deg p}]$ are linearly dependent.
Therefore, the following result is quite natural.

\begin{thm}\label{thm:2} (cf. Thm. 3.7.23 in~\cite{kreuzer00}). Let $P=K[x,y_1,\dots,y_n]$
  and let $I\subseteq P$ be a radical ideal of dimension zero.
  Let $p$ be a generator of $I\cap K[x]$. Then
  the following conditions are equivalent:
  \begin{enumerate}
  \item $I$ is in normal position w.r.t. $x$
  \item $\deg p=\dim_K P/I$
  \item $K[x]/\<p>$ and $P/I$ are isomorphic as $K$-algebras.
  \end{enumerate}
\end{thm}

Finally, the shape lemma can be stated as follows.

\begin{thm}\label{thm:3} (Shape Lemma; cf. Thm. 3.7.25 in~\cite{kreuzer00}) Let $P=K[x,y_1,\dots,y_n]$
  and let $I\subseteq P$ be a radical ideal of dimension zero that is in normal position
  w.r.t.~$x$. Let $p$ be a generator of $I\cap K[x]$. Then
  there are polynomials $q_1,\dots,q_{n-1}\in K[x]$ with $\deg(q_i)<\deg(p)$
  for all $i$ such that
  $\{y_1-q_1,\dots,y_n-q_n, p\}$ is a basis of~$I$.
\end{thm}

Here and elsewhere, by a ``basis'' of an ideal we understand just a set of generators, not
necessarily minimal or independent in any sense.

The purpose of this note is to extend these well-known facts from commutative polynomial
rings to rings of differential operators. This is motivated by recent developments in the
area of symbolic integration for so-called D-finite functions~\cite{kauers23c}. Given such a
function $f(x,y)$, the goal is to evaluate a definite integral
\[
  F(x) = \int_\Omega f(x,y)\,dy.
\]
More precisely, given an ideal of annihilating operators for~$f(x,y)$, we want to compute
an ideal of annihilating operators for the integral~$F(x)$. A general approach to this
problem is known as creative telescoping~\cite{zeilberger90,zeilberger91,koutschan13,chyzak14}
% TODO add chen25 once it is published
and has been subject of intensive research
during the past decades. There are several algorithms for creative telescoping, some of
which assume that the ideal of operators for $f(x,y)$ has a basis of the form $\{D_y-M,L\}$,
where $M$ and $L$ are operators in $D_x$ only. Thanks to the shape lemma, this is a fair
assumption.

The technique of creative telescoping also applies to summation problems. In this case, we
have to deal with recurrence operators rather than differential operators. It would be
interesting to have a version of the shape lemma also in this case. However, our results
for the differential case do not seem to extend easily to the recurrence case.

\section{Differential operators}

The role of the field~$K$ in the commutative setting sketched
in the introduction is now taken over by the field $C(x,y_1,\dots,y_n)$
of rational functions in~$x$ and~$y_1,\dots,y_n$, with coefficients in some
constant field~$C$ that we assume to have characteristic zero.
Hence from now on, we let $K=C(x,y_1,\dots,y_n)$. We will also abbreviate
$y_1,\dots,y_n$ by~$\vec y$.

We use the symbols $D_x$ and $D_{y_1},\dots,D_{y_n}$ to denote the partial derivation
operators, i.e., $D_x(f)=\frac{\partial f}{\partial x}$ and
$D_{y_i}(f)=\frac{\partial f}{\partial y_i}$ ($i=1,\dots,n$).
Note that $D_x(c)=D_{y_i}(c)=0$ for all $c\in C$ and all~$i$.

The action of $D_x$ and $D_{y_1},\dots,D_{y_n}$ turns $K$ into a partial differential field.
In general, if $L$ is a field, a map $D\colon L\to L$ is called a \emph{derivation} if
$D(a+b)=D(a)+D(b)$ and $D(ab)=D(a)b+aD(b)$ for all $a,b\in L$.
We call $L$ a partial differential extension field of $K$ if it is equipped with $n+1$
derivations that agree with the action of $D_x$ and $D_{y_1},\dots,D_{y_n}$ on the subfield~$K$.

Let $K[D_x,\vec{D_y}]:=K[D_x,D_{y_1},\dots,D_{y_n}]$ denote the ring of linear
differential operators with rational function coefficients, i.e.,
\[
  K[D_x,\vec{D_y}] = \biggl\{ \sum_{i,j_1,\dots,j_n=0}^d a_{i,j_1,\dots,j_n}(x,\vec y) D_x^i D_{y_1}^{j_1}\cdots D_{y_n}^{j_n}
  \mathrel{\bigg|} d\in\set N, a_{i,j_1,\dots,j_n} \in K \biggr\}.
\]
Because of the product rule, we have the commutation rules
$D_x\cdot x=x\cdot D_x+1$ and $D_{y_i}\cdot y=y\cdot D_{y_i}+1$ for every~$i$, so the
ring $K[D_x,\vec{D_y}]$ is non-commutative. A linear partial differential
equation can then be written as $L(f)=0$ with $L\in K[D_x,\vec{D_y}]$.

Let $C[[x]]$ and $C[[x,\vec y]]$ denote, as usual, the rings of univariate
and multivariate formal power series with coefficients in~$C$, and let $C((x))$ and
$C((x,\vec y))$ denote their respective quotient fields.

Let
$L=\sum_{i=0}^r a_i(x)D_x^i \in C(x)[D_x]$ be a linear ordinary differential
operator. An element $x_0\in C$ is called a regular point (or ordinary point)
of~$L$ if $a_r(x_0)\neq0$ and $a_i(x_0)$ is defined for all $0\leq i\leq r$, i.e., if no
coefficient~$a_i$ has a pole at~$x_0$. Via the change of variables
$x\leftarrow x-x_0$ the point~$x_0$ can be moved to the origin. Hence, without
loss of generality, assume that $0$ is a regular point of~$L$. Then the set of
power series solutions
\[
  V(L)=\{ f \in C[[x]] \mid L(f)=0 \}
\]
forms a $C$-vector space of dimension~$r$.

For a power series~$f(x,\vec y)\in C[[x,\vec y]]$, we define the ($K[D_x,\vec{D_y}]$-)
annihilator of~$f$ as the set of all operators that annihilate~$f$,
that is $\{L\in K[D_x,\vec{D_y}] \mid L(f)=0\}$. It is easily verified
that this set forms a (left) ideal in $K[D_x,\vec{D_y}]$. The series~$f$ is
called D-finite if $\dim_K(K[D_x,\vec{D_y}]/I)<\infty$, where $I$ denotes
the annihilator of~$f$. Equivalently, $f$ is called D-finite if $I$
is an ideal of dimension zero.

Also in the multivariate setting we can make
a similar statement about the dimension of the solution space, which
directly follows from Thm.~3.7 in~\cite{chen18}.

\begin{thm}\label{thm:solspace}
  Let $I$ be a zero-dimensional left ideal of $K[D_x,\vec{D_y}]$ 
  and $r=\dim_K(K[D_x,\vec{D_y}]/I)\in\set N$.
  If $0$ is an ordinary point of~$I$, then the set
  \[
  V(I)=\{f\in C[[x,\vec y]]\mid\forall L\in I:L(f)=0\}
  \]
  is a $C$-vector space of dimension~$r$.
\end{thm}

The definition of ordinary points proposed in \cite{chen18} is a bit more
complicated than the definition in the univariate case. We won't need it
here, so we do not reproduce it. It suffices to know that almost every
point is ordinary, so if $0$ is not a ordinary point, we always have
the option to get into the situation of Thm.~\ref{thm:solspace} by making
a change of variables.

For $f_1,\dots,f_r\in C((x,\vec y))$, their
\emph{Wronskian} (with respect to the variable~$x$) is denoted and
defined as follows:
\[
  \Wr_x(f_1,\dots,f_r) = \det
  \begin{pmatrix}
    f_1 & f_2 & \cdots & f_r \\
    D_x(f_1) & D_x(f_2) & \cdots & D_x(f_r) \\
    \vdots & \vdots & \ddots & \vdots \\
    D_x^{r-1}(f_1) & D_x^{r-1}(f_2) & \cdots & D_x^{r-1}(f_r)
  \end{pmatrix}.
\]
The Wronskian $\Wr_x(f_1,\dots,f_r)$ is equal to zero if and only
if the $f_i$ satisfy a linear relation with coefficients that do
not depend on~$x$, e.g., if $\sum_{i=1}^r a_i f_i = 0$ with
$a_i\in C((\vec y))$ not all zero~\cite{bocher1901}.

For later use, we state the following lemma.

\begin{lem}\label{lem:elim}
  If $L$ is a partial differential field extension of $K$ and $I$ is an ideal in $L[D_x,\vec{D_y}]$
  which has a basis in $K[D_x,\vec{D_y}]$, then also the elimination ideal
  $I \cap L[D_x]$ has a basis in~$K[D_x]$.
\end{lem}
\begin{proof}
  Let $P_1,\dots,P_m\in K[D_x,\vec{D_y}]$ be a basis of~$I$, and let $M$
  be an element in the elimination ideal $I\cap K[D_x]$. Then there
  exist $Q_1,\dots,Q_m\in L[D_x,\vec{D_y}]$ such that $M=Q_1P_1+\dots+Q_mP_m$.

  Clearly, $L$ can be viewed as a $K$-vector space, of potentially
  infinite dimension. In any case, there exists a finite-dimensional
  $K$-subspace~$V$ of $L$ that contains all the coefficients of the~$Q_i$
  (note that each $Q_i$ has only finitely many coefficients in~$L$).
  Now let $B_1,\dots,B_d$ be a $K$-basis of~$V$, which means that
  there are $Q_{i,j}\in K[D_x,\vec{D_y}]$ such that
  $Q_i=Q_{i,1}B_1+\dots+Q_{i,d}B_d$ for all~$i$. Hence we can write
  \begin{equation}\label{eq:Qij}
    M = \sum_{i=1}^m\biggl(\sum_{j=1}^d Q_{i,j}B_j\biggr)P_i
      = \sum_{j=1}^d\biggl(\sum_{i=1}^m Q_{i,j}P_i\biggr)B_j.
  \end{equation}

  Since the $B_j$ are linearly independent over~$K$, it follows that for
  each~$j$, the quantity $\sum_{i=1}^m Q_{i,j}P_i$ is free of~$D_{y_1},\dots,D_{y_n}$, because
  $M$ is free of $D_{y_1},\dots,D_{y_n}$ and because there cannot be a cancellation on the
  right-hand side of~\eqref{eq:Qij}. Therefore, the coefficients
  $\sum_{i=1}^m Q_{i,j}P_i$ are in $K[D_x]$, which proves the claim.
\end{proof}

For readers familiar with the theory of Gr\"obner bases, we offer the
following alternative proof: from a given basis of $I$ with elements
in $K[D_x,\vec{D_y}]$, we obtain a basis of $I\cap L[D_x]$ by computing a
Gr\"obner basis with respect to an elimination order. Since Buchberger's
algorithm never extends the ground field, the resulting basis must
be a subset of~$K[D_x]$.

\section{The Shape Lemma}

For an ideal $I\subseteq K[D_x,\vec{D_y}]$ of dimension zero, consider the
quotient $K[D_x,\vec{D_y}]/I$ as a $K[D_x]$-module. Since its dimension
as $K$-vector space is finite, this module must be cyclic~\cite[Prop. 2.9]{put03}.
If $M\in K[D_x,\vec{D_y}]$ is such that the equivalence class $[M]$ is a
generator of the module, then there is an $L\in K[D_x]$ such
that $L\cdot [M]=[LM]=[1]$.
Therefore, evaluating an integral
\[
  F(x)=\int_\Omega f(x,\vec y)\,d\vec y
\]
for a function $f(x,\vec y)$ whose ideal of annihilating operators is~$I$ is the
same as evaluating the integral
\[
  F(x)=\int_\Omega L\cdot g(x,\vec y)\,d\vec y
\]
where $g(x,\vec y)$ is defined as $M(f(x,\vec y))$. The choice of $M$ implies
that the annihilating ideal $J$ of $g(x,\vec y)$ has a basis of the form $\{D_{y_1} - Q_1,\dots, D_{y_n}-Q_n, P\}$
for some operators $P,Q_1,\dots,Q_n$ in $K[D_x]$.

Transforming $I$ to $J$ is known as gauge transform and can be considered as a
satisfactory solution to our problem: every ideal $I\subseteq K[D_x,\vec{D_y}]$ of dimension
zero can brought to an ideal $J$ to which the shape lemma applies by means of a
gauge transform. 

We shall propose an alternative approach here. Rather than applying a gauge transform,
which amounts to applying an operator to the integrand, our question is whether we
can also obtain an ideal basis of the required form by applying a linear change of
variables, i.e., using
\[
  F(x) = \int_\Omega f(x,\vec y)\,d\vec y = \int_{\tilde\Omega} f(x,y_1+c_1x,\dots,y_n+c_nx)\,d\vec y
\]
for some constants $c_1,\dots,c_n$ (and an appropriately adjusted integration range). It turns out
that this perspective leads to a shape lemma for differential operators that matches
more closely the situation in the commutative case. 

Note that
\[
L(x,y_1,\dots,y_n,D_x,D_{y_1},\dots,D_{y_n})\in K[D_x,D_y]
\]
is an annihilating operator of $f(x,\vec y+\vec cx)=f(x,y_1+c_1x,\dots,y_n+c_nx)$ if and only if
\[
L(x,y_1-c_1x,\dots,y_n-c_nx,D_x+c_1D_{y_1}+\cdots+c_nD_{y_n},D_{y_1},\dots,D_{y_n})
\]
is an annihilating operator of~$f(x,\vec y)$. In particular, the ideal
of annihilating operators of $f(x,\vec y)$ has dimension zero if and only this is the case for
the ideal of annihilating operators of $f(x,\vec y+\vec cx)$.

We shall show (Thm.~\ref{thm:changofvariables} below) that every
zero-dimensional left ideal of $K[D_x,\vec{D_y}]$ can be brought to normal position
by a change of variables $\vec y\leftarrow\vec y+\vec cx$. For the notion of being in
normal position, we propose the following definition.

\begin{defi}\label{def:normal}
  Let $I\subseteq K[D_x,\vec{D_y}]$ be an ideal of dimension zero,
  so that $r=\dim_K K[D_x,\vec{D_y}]/I$ is finite.
  The ideal $I$ is called \emph{in normal position} (w.r.t.~$D_x$) if for every
  choice of $C$-linearly independent solutions $f_1,\dots,f_r$
  we have $\Wr_x(f_1,\dots,f_r)\neq0$.
\end{defi}

\begin{ex}
  For the ideal $I=\<(D_x-1)(D_x-2),D_y>$ we have $r=2$. The
  solution space of $I$ is generated by $\exp(x)$ and~$\exp(2x)$.
  We have $\Wr_x(\exp(x),\exp(2x))=\exp(3x)$. Therefore, $I$ is
  in normal position w.r.t. $D_x$.
  However, as $D_y(\exp(x))=D_y(\exp(2x))=0$, we also
  have $\Wr_y(\exp(x),\exp(2x))=0$, so $I$ is \underline{not} in
  normal position w.r.t.~$D_y$.
\end{ex}

With this notion of being in normal position, we can state the following result. 

\begin{thm} (Shape Lemma; differential analog of Thms.~\ref{thm:2} and~\ref{thm:3})\label{thm:our23}
  Let $I\subseteq K[D_x,\vec{D_y}]$ be an ideal of dimension zero.
  Let $P$ be a generator of $I\cap K[D_x]$. Then the following conditions are equivalent:
  \begin{enumerate}
  \item $I$ is in normal position w.r.t. $D_x$
  \item $\ord(P)=\dim_K K[D_x,\vec{D_y}]/I$
  \item $K[D_x]/\<P>$ and $K[D_x,\vec{D_y}]/I$ are isomorphic as $K[D_x]$-modules.
  \item There are $Q_1,\dots,Q_n\in K[D_x]$ with $\ord(Q_i)<\ord(P)$ for all $i$
    such that $\{D_{y_1}-Q_1,\dots,D_{y_n}-Q_n,P\}$ is a basis of~$I$.
  \end{enumerate}
\end{thm}
\begin{proof}
  Let $r=\dim_K K[D_x,\vec{D_y}]/I$.

  $1\Rightarrow 2$\quad
  To show that $\ord(P)=r$, suppose that $\ord(P)<r$ and let $f_1,\dots,f_r$ be
  some $C$-linearly independent solutions of~$I$. By Thm.~\ref{thm:solspace},
  we may assume that such solutions exist. 
  As no more than $\ord(P)$ solutions of $P$ can be linearly independent over
  $C[[\vec y]]$, it follows that $f_1,\dots,f_r$ are linearly dependent over~$C[[\vec y]]$.
  This implies $\Wr_x(f_1,\dots,f_r)=0$, in contradiction to the assumption that
  $I$ is in normal position.

  $2\Rightarrow 1$\quad
  Let $f_1,\dots,f_r$ be $C$-linearly independent solutions of~$I$.
  We have to show that they are also linearly independent over~$C((\vec y))$.
  Suppose otherwise. Then we may assume that $f_r$ is a $C((\vec y))$-linear combination
  of $f_1,\dots,f_{r-1}$. The operator
  \[
  Q = \begin{vmatrix}
    f_1 & \cdots & f_{r-1} & 1\\
    D_x(f_1) & \cdots & D_x(f_{r-1}) & D_x \\
    \vdots & & \vdots & \vdots \\
    D_x^{r-1}(f_1) & \cdots & D_x^{r-1}(f_{r-1}) & D_x^{r-1}
    \end{vmatrix} \in C((x,\vec y))[D_x]
  \]
  has the solutions $f_1,\dots,f_{r-1}$ and~$f_r$.
  It must therefore belong to the ideal generated by $I$ in the larger ring $C((x,\vec y))[D_x,\vec{D_y}]$,
  for if it didn't, then $\dim_{C((x,\vec y))} C((x,\vec y))[D_x,\vec{D_y}]/(\<I>+\<Q>)<r$, which
  is impossible when the solution space has $C$-dimension~$r$.

  By Lemma~\ref{lem:elim}, $P$~is also a generator of the elimination ideal
  $\<I>\cap C((x,\vec y))[D_x]$, where $\<I>$ denotes the ideal generated by $I$
  in $C((x,\vec y))[D_x,\vec{D_y}]$. By assumption we have $\ord(P)=r>\ord(Q)$. This is
  a contradiction. 
  
  $2\Rightarrow 3$\quad
  Consider the function $\phi\colon K[D_x]/\<P>\to K[D_x,\vec{D_y}]/I$ defined by
  $\phi([L]_{\<P>}):=[L]_I$.
  This function is well-defined because $\<P>\subseteq I$.
  The function is obviously a morphism of $K[D_x]$-modules, and it is injective,
  because if $L\in K[D_x]$ is such that $[L]_I=[0]_I$, then $L\in I$,
  so $L\in I\cap K[D_x]=\<P>$, so $[L]_{\<P>}=0$.
  Being a morphism of $K[D_x]$-modules, $\phi$ is in particular a morphism of $K$-vector spaces.
  Therefore, since $\dim_K K[D_x]/\<P>=r=\dim_K K[D_x,\vec{D_y}]/I$ by assumption,
  injectivity implies bijectivity, and therefore $\phi$ is an isomorphism.

  $3\Rightarrow 2$\quad clear.

  $2\Rightarrow 4$\quad
  By assumption, the elements $[1],[D_x],\dots,[D_x^{r-1}]$ of $K[D_x,\vec{D_y}]/I$ are $K$-linearly
  independent and therefore form a vector space basis of $K[D_x,\vec{D_y}]/I$.
  Therefore, the element $[D_y]$ of $K[D_x,\vec{D_y}]/I$ can be expressed as a
  $K$-linear combination of $[1],[D_x],\dots,[D_x^{r-1}]$. This implies
  the existence of a~$Q$.

  $4\Rightarrow 2$\quad 
  By repeated addition of suitable multiples of basis elements, it can be seen
  that every element of $K[D_x,\vec{D_y}]$ is equivalent modulo $I$ to an element
  of the form $q_0+q_1D_x+\cdots+q_{r-1}D_x^{r-1}$. Therefore, the elements $[1],\dots,[D_x^{r-1}]$
  generate $K[D_x,\vec{D_y}]/I$ as a $K$-vector space. This implies
  $\dim_K K[D_x,\vec{D_y}]/I\leq r$.
  At the same time, the dimension cannot be smaller than~$r$, because if
  $[1],\dots,[D_x^{r-1}]$ were $K$-linearly dependent, then $I\cap K[D_x]$ would
  contain an element of order less than $\ord(P)$, which is impossible by the choice of~$P$.
\end{proof}

Again, readers familiar with the theory of Gr\"obner bases will have no difficulty
finding shorter arguments for some of the implications. 

The similarity of Thm.~\ref{thm:our23} to the corresponding theorems for
commutative polynomial rings is evident, but there are some subtle differences
as well. One difference is that Thms.~\ref{thm:2} and~\ref{thm:3} require the ideal
to be radical, while no such assumption is needed for Thm.~\ref{thm:our23}.

However, it turns out that in order to also generalize Thm.~\ref{thm:1} to
differential operators, we do need to introduce a restriction. Note that
Thm.~\ref{thm:1} becomes wrong for non-radical ideals if we interpret their
solutions as points with multiplicities. Indeed, in this sense, a non-radical
ideal is never in normal position, and no linear change of variables will
suffice to turn a non-radical ideal into a radical ideal.

Ideals of differential operators cannot have multiple solutions (cf.~Thm.~\ref{thm:solspace}).
Instead, it seems appropriate to adopt the following notion. 

\begin{defi}\label{def:D-radical}
  A finite dimensional $C$-vector space $V$ is called \emph{linearly disjoint}
  with~$K$ (over~$C$) if $\dim_{K} K \otimes_{C} V = \dim_{C} V$, or equivalently,
  if any $C$-basis of $V$ is $K$-linearly independent. A zero-dimensional ideal
  $I \subseteq K[D_{x},\vec{D_{y}}]$ is called \emph{D-radical} if its solution
  space~$V(I)$ is linearly disjoint with~$K$.
\end{defi}

Observe the difference between Defs.~\ref{def:D-radical} and~\ref{def:normal}.
In both cases we require the absence of linear relations, but with respect to
different coefficient domains. For normal position, the coefficients must be
free of $x$ but can depend in an arbitrary way on~$\vec y$, and for D-radical
the coefficients must be rational functions in $x$ and~$\vec y$.

\begin{ex}\label{ex:D-radical}\
  \begin{enumerate}
  \item Let $I=\<(D_x-1)^2,D_y> \subseteq K[D_{x},D_{y}]$, then $V(I)$ contains the
    $C$-linearly independent solutions $\exp(x)$ and $x\exp(x)$.
    As these are not linearly independent over~$K$, the ideal~$I$ is not D-radical.
  \item The solution space of the ideal $\<(D_x-1)(D_x-2),D_y>$ has the basis
    $\{\exp(x),\exp(2x)\}$. Since $\exp(x)$ and $\exp(2x)$ are linearly independent
    over $K=C(x)$, the ideal is D-radical.
  \end{enumerate}
\end{ex}

Note that in both instances of Example~\ref{ex:D-radical} the generators of
the ideal actually lie in the commutative ring $C[D_{x},D_{y}]$. We observe
that the corresponding ideal in $C[D_{x},D_{y}]$ is radical (in the
commutative sense) in Case~(2), but not radical in Case~(1). This is not a
coincidence.  Our definition of D-radicality specializes to the classical
concept of radicality when differential operators with constant coefficients
are considered, which justifies the choice of the name. More precisely: for a
zero-dimensional ideal $I \subseteq C[D_{x},\vec{D_{y}}]$ we have that $I$ is
D-radical if and only if it is radical (in the commutative sense). This can be
understood by looking at the closed-form solutions of such
constant-coefficient differential equations: the solution space of the
operator $(D_x-\alpha_1)\cdots(D_x-\alpha_r)$ is spanned by
$\exp(\alpha_1x),\dots,\exp(\alpha_rx)$, which are linearly independent over~$K$,
whenever the $\alpha_i$ are pairwise disjoint.
In contrast, the solution space of the operator $(D_x-\alpha)^r$ is spanned
by $\exp(\alpha x),x\exp(\alpha x),\dots,x^{r-1}\exp(\alpha x)$, which are
clearly linearly dependent over~$K$. This argument applies analogously to the
situation of several variables.

The correspondence between radical and D-radical
also extends to Theorem~\ref{thm:changofvariables} below, which reduces to
Theorem~\ref{thm:1} for ideals generated by operators with constant coefficients.
In particular, the change of variable $\vec{y} \leftarrow \vec{y} + \vec{c}x$ keeps $D_{y_{1}}, \dots, D_{y_{n}}$ unchanged, and replaces $D_{x}$ with $D_{x} + c_{1}D_{y_{1}} + \dots + c_{n}D_{y_{n}}$, and thus the theorem yields a change of variable with the same structure as Prop.~\ref{thm:1}.

\begin{thm} (Differential analog of Thm.~\ref{thm:1})\label{thm:changofvariables}
  Let $I\subseteq K[D_x,\vec{D_y}]$ be a zero-dimensional D-radical ideal.
  Then there are constants $c_1,\dots,c_n\in C$ such that the ideal $J$ obtained
  from $I$ by applying the linear change of variables $\vec y\leftarrow\vec y+\vec cx$
  (where $\vec c=(c_1,\dots,c_n)$) is in normal position w.r.t.~$D_x$.
\end{thm}
\begin{proof}
  We show that whenever $f_1(x,\vec y),\dots,f_r(x,\vec y)$ form a $C$-basis of $V(I)$ such that 
  $\Wr_x\bigl(f_i(x,\vec y+\vec cx)\bigr)_{i=1}^r=0$ for all $\vec c\in C^n$,
  then $f_1,\dots,f_r$ are $K$-linearly dependent.
  Since the ideal is D-radical, this implies that $f_1,\dots,f_r$ are
  $C$-linearly dependent, which is a contradiction.

  Consider $c_1,\dots,c_n$ as an additional variables and recall that the
  assumption $\Wr_x\bigl(f_i(x,\vec y+\vec cx)\bigr)_{i=1}^r=0$
  implies that the $f_i(x,\vec y+\vec cx)$ are linearly dependent over the constant
  field with respect to~$x$, i.e., $C((\vec y,\vec c))$-linearly dependent: thus we can
  assume that there exist $p_1,\dots,p_r\in C((\vec y,\vec c))$, not all $0$, such that
  \begin{equation}\label{eq:lincomb_fi}
    \sum_{i=1}^r p_i(\vec y,\vec c)\cdot f_i(x,\vec y+\vec cx) = 0.
  \end{equation}
  Each $f_i$ has an expansion as a series in~$x$:
  \[
    f_i(x,\vec y+\vec cx) = \sum_{j=0}^\infty \ \underbrace{\frac{1}{j!}
      \frac{\partial^j f_i(x,\vec y+\vec cx)}{\partial x^j}\bigg|_{x=0}}_%
    {\textstyle =: f_{i,j}(\vec y,\vec c)}
    \cdot\ x^j.
  \]
  Note that the series coefficients~$f_{i,j}$ are polynomials in~$\vec c$, because
  \[
    \frac{\partial^j f_i(x,y+cx)}{\partial x^j}\bigg|_{x=0} =
    \sum_{k_0+\cdots+k_n=j} \binom{j}{k_0,\dots,k_n}
    \cdot \underbrace{\biggl[
        f_i^{(k_0,\dots,k_n)}(x,\vec y+\vec cx)\biggr]_{x=0}}_{\in C((\vec y))}
    \cdot c_1^{k_1}\cdots c_n^{k_n}
    \in C((\vec y))[\vec c].
  \]
  Here the notation $f_i^{(k_0,\dots,k_n)}$ refers to $k_0$-fold derivative w.r.t. the first argument, the $k_1$-fold
  derivative with respect to the second argument, etc.
  
  It follows that Eq.~\eqref{eq:lincomb_fi} can be expanded as 
  \[
    \sum_{j=0}^\infty \left(
      \sum_{i=1}^r p_i(\vec y,\vec c) \cdot f_{i,j}(\vec y,\vec c) \right) \cdot x^j = 0
  \]
  and therefore, for all $j \in \set N$, $\sum_{i=1}^r p_i(\vec y,\vec c) \cdot f_{i,j}(\vec y,\vec c)=0$.

  Let $M$ be the matrix $\bigl(f_{i,j}(\vec y,\vec c)\bigr)_{j\geq 0,1\leq i\leq r}$
  with infinitely many rows and $r$ columns.
  From the above, 
  \[
    \bigl(p_i(\vec y,\vec c)\bigr)_{i=1}^r \in \ker M,
  \]
  and therefore $M$ is rank-deficient; let $R<r$ denote the rank of~$M$. Hence
  there exists an integer $m\in\set{N}$ such that the rank of the
  $(m\times r)$-submatrix~$M'$, that is obtained by taking the first $m$ rows
  of~$M$, is also equal to~$R$. It follows that $\ker(M')=\ker(M)$, and since
  $M'\in C((\vec y))[\vec c]^{m\times r}$ we have that $\ker(M')$ is a subspace of
  $C((\vec y))(\vec c)^r$. Therefore, the coefficients $p_i(\vec y,\vec c)$ can be chosen in
  $C((\vec y))(\vec c)$. In fact, by clearing denominators, we can even assume them to
  belong to $C[[\vec y]][\vec c]$.

  Now perform the substitution $\vec c\leftarrow\vec c-\vec y/x$ in~\eqref{eq:lincomb_fi} to get
  \begin{equation}
    \label{eq:4}
    \sum_{i=1}^r p_i(\vec y,\vec c-\vec y/x)\cdot f_i(x,\vec cx) = 0.
  \end{equation}
  Each $p_i(\vec y,\vec c-\vec y/x)$ admits an expansion as a power series in $y_1,\dots,y_n$
  \begin{equation}
    \label{eq:2}
    p_{i}(\vec y, \vec c-\vec y/x) = \sum_{j_1,\dots,j_n=0}^\infty q_{i,j_1,\dots,j_n}(\vec c,x) y_1^{j_1}\cdots y_n^{j_n}.
  \end{equation}
  Eq.~\eqref{eq:4} then expands as
  \begin{equation}
    \label{eq:10}
    \sum_{j_1,\dots,j_n=0}^\infty\big(\sum_{i=1}^r q_{i,j}(\vec c,x)\cdot f_i(x,\vec cx)\big) y_1^{j_1}\cdots y_n^{j_n} = 0
  \end{equation}
  and therefore, for all $j_1,\dots,j_n\in\set Z$, 
  \begin{equation}
    \label{eq:12}
    \sum_{i=1}^r q_{i,j}(\vec c,x)\cdot f_i(x,\vec cx) = 0.
  \end{equation}
  Since the $p_i$ are not all~$0$, there must exist $i,j_1,\dots,j_n$ with $q_{i,j_1,\dots,j_n}\neq 0$,
  and therefore for such a choice of $j_1,\dots,j_n$, the left-hand side of Eq.~\eqref{eq:12} is
  a non-trivial linear combination.
  
  Furthermore, observe that since the $p_i$ are rational in their second
  argument, the coefficients $q_i$ are rational functions.
  So finally, substituting $\vec c\leftarrow\vec y/x$ yields the desired dependency with coefficients in~$K$:
  \[
    \sum_{i=1}^r q_i(\vec y/x,x) f_i(x,\vec y) = 0.\qedhere
  \]
\end{proof}

\begin{ex}
  The annihilator~$I_1$ of $\exp(x)$ and $y\exp(x)$ is not D-radical. The annihilator~$I_2$
  of $\exp(x)$ and $\exp(x+y)$ is D-radical but not in normal position w.r.t.~$D_x$.
  Setting $y$ to $y+cx$ in~$I_1$ gives the annihilator of $\exp(x),(y+cx)\exp(x)$,
  which is still not D-radical. However, setting $y$ to $y+cx$ in~$I_2$ gives
  the annihilator of $\exp(x),\exp((1+c)x+y)$, which is in normal
  $\partial_x$-position for every choice $c\neq0$.
\end{ex}

Recall that our motivation was the computation of an ideal of annihilating operators for an integral
\begin{equation*}
  \int_{\Omega}  f(x,\vec{y})\,d\vec{y}
\end{equation*}
where $f(x,\vec y)$ is a D-finite function.
Let $I \subseteq K[D_{x}, \vec{D_{y}}]$ be the annihilating ideal of $f(x,\vec{y})$, and assume that it is D-radical.
According to Theorem~\ref{thm:changofvariables}, there exists $\vec{c}$ such that, after the change of
variables $\vec{y} \leftarrow \vec{y} + \vec{c}x$, the ideal $I$ is in normal position.
According to the Shape Lemma~\ref{thm:our23}, this implies that, after change of variables, the ideal $I$ is generated
by $P$, and $D_{y_1}-Q_1,\dots,D_{y_n}-Q_n$, for certain $P, Q_1,\dots,Q_n\in K[D_{x}]$.
This means that $g(x,\vec{y}) = f(x,\vec{y} + \vec{c}x)$ is such that $P(g)=0$
and $D_{y_i}(g)=Q_i(g)$ for all~$i$, and we can use creative telescoping to compute an annihilating ideal $J$ for 
\begin{equation*}
  \int_{\tilde{\Omega}}  g(x,\vec{y})\,d\vec{y} = \int_{\Omega} f(x,\vec{y})\,d\vec{y},
\end{equation*}
where $\tilde{\Omega}$ is the inverse image of $\Omega$ under the change of variables.

It remains open how these results extend to the recurrence case. While the restriction to the differential case does not
seem essential for the Shape Lemma itself (Theorem~\ref{thm:our23}), there is a substantial difference as far as the
effect of a linear change of variables $y_i\leftarrow y_i+c_ix$  on annihilating operators is concerned: While $D_x$ gets
replaced by $D_x+c_1D_{y_1}+\cdots+c_nD_{y_n}$ in the differential case, the shift operator $S_x$~would have to be replaced
by $S_xS_{y_1}^{c_1}\cdots S_{y_n}^{c_n}$. Even if we restrict $c_1,\dots,c_n$ to nonnegative integers in order to make this
meaningful, it is not clear how Theorem~\ref{thm:changofvariables} could be adapted to this situation. 

\medskip
\noindent\textbf{Acknowledgement.} We thank the anonymous referee for his or her valuable comments, in particular a suggested
change of notation that we initially were skeptic about but that indeed worked out more smoothly than we had expected. 

\bibliographystyle{plain}
\bibliography{bib}

\begin{thebibliography}{10}

\bibitem{becker93}
Thomas Becker, Volker Weispfenning, and Heinz Kredel.
\newblock {\em {G}r{\"o}bner Bases}.
\newblock Springer, 1993.

\bibitem{bocher1901}
Maxime B\^{o}cher.
\newblock Certain cases in which the vanishing of the {W}ronskian is a
  sufficient condition for linear dependence.
\newblock {\em Transactions of the American Mathematical Society}, 2:139--149,
  1901.

\bibitem{buchberger65}
Bruno Buchberger.
\newblock {\em Ein Algorithmus zum Auffinden der Basiselemente des
  Restklassenrings nach einem nulldimensionalen Polynomideal}.
\newblock PhD thesis, Universit{\"a}t Innsbruck, 1965.

\bibitem{buchberger10}
Bruno Buchberger and Manuel Kauers.
\newblock Gr{\"o}bner basis.
\newblock {\em Scholarpedia}, 5(10):7763, 2010.
\newblock \url{http://www.scholarpedia.org/article/Groebner_basis}.

\bibitem{chen18}
Shaoshi Chen, Manuel Kauers, Ziming Li, and Yi~Zhang.
\newblock Apparent singularities of {D}-finite systems.
\newblock {\em Journal of Symbolic Computation}, 95(10):217--237, 2019.

\bibitem{chyzak14}
Fr{\'e}d{\'e}ric Chyzak.
\newblock {\em The {ABC} of Creative Telescoping -- Algorithms, Bounds,
  Complexity}.
\newblock Habilitation {\'a} diriger des recherches. Universit{\'e} Paris-Sud
  11, 2014.

\bibitem{cox92}
David Cox, John Little, and Donal O'Shea.
\newblock {\em Ideals, Varieties, and Algorithms}.
\newblock Springer, 1992.

\bibitem{cox05}
David Cox, John Little, and Donal O'Shea.
\newblock {\em Using Algebraic Geometry}.
\newblock Springer, 2nd edition, 2005.

\bibitem{gianni89}
Patricia Gianni and Teo Mora.
\newblock Algebraic solutions of systems of polynomial equations using
  {G}r{\"o}bner bases.
\newblock In {\em Proceedings of the 5th International Conference on Applied
  Algebra, Algebraic Algorithms and Error-Correcting Codes}, volume 356 of {\em
  Lecture Notes in Computer Science}, pages 247--257, 1989.

\bibitem{kauers23c}
Manuel Kauers.
\newblock {\em D-Finite Functions}.
\newblock Springer, 2023.

\bibitem{koutschan13}
Christoph Koutschan.
\newblock Creative telescoping for holonomic functions.
\newblock In {\em Computer Algebra in Quantum Field Theory: Integration,
  Summation and Special Functions}, Texts and Monographs in Symbolic
  Computation, pages 171--194. Springer, 2013.

\bibitem{kreuzer00}
Martin Kreuzer and Lorenzo Robbiano.
\newblock {\em Computational Commutative Algebra~I}.
\newblock Springer, 2000.

\bibitem{put03}
Marius van~der Put and Michael Singer.
\newblock {\em {G}alois Theory of Linear Differential Equations}.
\newblock Springer, 2003.

\bibitem{zeilberger90}
Doron Zeilberger.
\newblock A holonomic systems approach to special functions identities.
\newblock {\em Journal of Computational and Applied Mathematics}, 32:321--368,
  1990.

\bibitem{zeilberger91}
Doron Zeilberger.
\newblock The method of creative telescoping.
\newblock {\em Journal of Symbolic Computation}, 11:195--204, 1991.

\end{thebibliography}

\end{document}